\documentclass{article}[11pt]
\usepackage{amssymb, amsthm, amsmath}
\usepackage[a4paper]{geometry}
\usepackage{amssymb,latexsym,amscd}
\usepackage{epsfig}
\usepackage[active]{srcltx}
\usepackage{psfrag}
\usepackage{graphics,graphicx}
\usepackage{braket}
\usepackage{color}
\usepackage{natbib}

\usepackage{pstricks,pst-node,pst-tree}

\usepackage{setspace}

\definecolor{r}{rgb}{1,0,0}
\definecolor{b}{rgb}{0,0,1}

 \newcommand{\tr}[1]{}

\newcommand{\sep}{\mathord{:}}

\newcommand{\fra}[2]{\textstyle{\frac{#1}{#2}}}

\newcommand{\beqn}{\begin{eqnarray}\begin{aligned}}
\newcommand{\eqn}{\end{aligned}\end{eqnarray}}
\newcommand{\lra}[1]{{\langle #1 \rangle}}
\newcommand{\lrs}[1]{{\left[ #1 \right]}}

\newcommand{\id}{\mathbf{1}}

\theoremstyle{plain}
\newtheorem{thm}{Theorem}[section]

\newtheorem{res}{Result}

\newtheorem{lem}[thm]{Lemma}

\theoremstyle{definition}

\title{A tensorial approach to the inversion of group-based phylogenetic models}
\author{Jeremy G. Sumner$^{1}$,  Peter D. Jarvis$^{3}$, and Barbara R. Holland$^{2}$}

\begin{document}
\maketitle

\begin{abstract}
Using a tensorial approach, we show how to construct a one-one correspondence between pattern probabilities and edge parameters for any group-based model.
This is a generalisation of the ``Hadamard conjugation'' and is equivalent to standard results that use Fourier analysis. 
In our derivation we focus on the connections to group representation theory and emphasize that the inversion is possible because, under their usual definition, group-based models are defined for abelian groups only.
We also argue that our approach is elementary in the sense that it can be understood as simple matrix multiplication where matrices are rectangular and indexed by ordered-partitions of varying sizes.
\end{abstract}

\vspace{5em}

\hrule\mbox{}\\
\thanks{\footnotesize{
\noindent
School of Mathematics and Physics, University of Tasmania, Australia\\
$^1$ARC Research Fellow\\
$^2$ARC Future Fellow\\
$^3$Alexander von Humboldt Fellow\\
\textit{keywords:} groups, representation theory, symmetry, Markov chains\\
\textit{email:} jsumner@utas.edu.au
}
}




\section{Introduction}

In a series of papers from 1989 and the early 90s, Hendy and colleagues introduced the Hadamard conjugation as a novel tool for phylogenetic analyses \citep{hendy1989,hendy1989b,hendy1993}.
They found an invertible relationship between a phylogenetic tree, as characterized by an edge length spectrum, and the probability of each site pattern (referred to as the sequence spectrum).
Originally introduced only for the 2-state symmetric model, the Hadamard conjugation was later extended to the K3ST model \citep{hendy1994} and further to any of the so-called ``group-based'' models \citep{szekely1993}.
Hadamard conjugation has been used as both a tool for simulation \citep{hendy1993c} and to look at statistical properties of methods, exploring the inconsistency of parsimony under a molecular clock \citep{hendy1993,holland2003}.
For these sorts of applications, following the notation in \citet{felsenstein2004}, we can use the Hadamard transform $H$ to start with an edge length spectrum $\gamma$ and calculate the sequence spectrum $s=H^{-1}\log(H\gamma)$.
The beauty of Hadamard conjugations is that one can also begin with an observed sequence spectrum $\hat{s}$ and perform the inverse of the conjugation to empirically obtain an edge length spectrum $\hat{\gamma} = H^{-1}\log(H \hat{s})$.
Although it is not expected that the $\hat{\gamma}$ spectrum will precisely match a tree, \citet{hendy1991} proposed using a optimisation criterion to map from $\hat{\gamma}$ to the ``closest tree''.

Several authors have commented that it is potentially a useful feature of Hadamard conjugation that data isn't forced onto a fixed tree.   
The conflicting information can be retained and interpreted in the form of a ``lentoplot'' \citep{lento1995}  or a splits-graph \citep{huber2001}, with both of these methods implemented in \emph{Spectronet} \citep{huber2002}.
\citet{schliep2009} gives some more statistical justification for such an approach by making a link to modern statistical techniques such as the Lasso and Ridge regression.

\citet{vonhaeseler1993} seems to be the first paper that explicitly suggests using Hadamard conjugation to provide a likelihood framework for networks.
The chief idea being that one can start with an edge length spectrum that encodes a set of incompatible splits, use the Hadamard transformation to get site probabilities and use these to determine a likelihood.
This idea was further explored by \citet{bryant2005}, and \citet{bryant2009} followed this through defining the ``$n$-taxon process'' for group-based models.
It should be noted that likelihoods calculated via Hadamard are not equivalent to likelihoods calculated by taking a mixture of trees.
Indeed, \citet{matsen2007,matsen2008} used Hadamard methods in combination with phylogenetic invariants to show that mixtures of trees with the same topology can exactly mimic another tree under the 2-state model.
Considering biological applications, thinking in terms of mixtures of trees or partitions where the data can be thought of as arising on a set of trees \citep{griffiths1996,griffiths1997,jin2006} seems more reasonable than the Hadamard conjugation.
\citet{strimmer2000} suggested using split networks as a spring board to likelihood-based analyses on DAGs, but later identified several problems with the approach \citep{strimmer2001}; most notably, in split-networks internal nodes do not have a biological interpretation as an ancestor.

In \citet{sumner2010}, we gave some additional insight into the interpretation of applying the Hadamard conjugation in a network setting.
We showed that permutation group structure inherent to the Hadamard transformation -- as for any group-based model -- restricts the resulting process from being capable of reproducing truly convergent processes.
This is a serious limitation, as one of the biological motivations for explicit network models is the ability to model convergent processes.
We also presented an alternative algebraic formalism for the general Markov model, analogous to the $n$-taxon process, but capable of reproducing convergent processes.
From the point of view of group representation theory, the inversion of group-based models relies on the fact that the \emph{irreducible} representations of an abelian group are one-dimensional, and the model structure essentially reduces to group characters -- hence the standard presentation of a Fourier inversion.
In this article, we make this connection concrete.
For the general Markov model, it is then immediately apparent that an analogous inversion is not possible because the underlying irreducible representations are not one-dimensional.
In fact, to obtain one-dimensional representations for the general Markov model, it is necessary to apply higher-degree polynomial maps (beyond the degree 1, linear case), and define ``Markov invariants'' \citep{sumner2008}.
These invariants present one-dimensional representations but at the cost of the higher degree -- degree 5 in the case of the general Markov model on four states on quartet trees \citep{sumner2009,holland2012}. 
This connection between Hadamard transformation and Markov invariants is an interesting one, but we do not discuss it further here.

In this paper we approach the inversion of group-based phylogenetic models by taking a representation-theoretic perspective and working explicitly with tensor indices.
Our approach rests heavily on the formalism of ``phylogenetic tensors'', as presented in \citet{bashford2004}, for the binary-symmetric and K3ST model, and \citet{sumner2008,sumner2010}, for the general Markov model.

\section{Background}

In this paper we consider the continuous-time formulation of Markov processes, and show how to implement the inversion of a group-based phylogenetic model based on \emph{any} abelian group.
We note that such an inversion requires a map from tensor product space (where elements are indexed by ordered-$n$-partitions) to phylogenetic splits (where elements are indexed by bipartitions).
We achieve this by finding canonical maps from bipartitions to ordered-$n$-partitions.

For a group $G$ with order $|G|=d$, we write $G=\{\sigma_1,\sigma_2,\ldots ,\sigma_d\}$, and, when necessary, write $\epsilon\in G$ to specify the identity element of $G$. 
Consider the vector space $\mathbb{C}^d\cong \left\langle G\right\rangle_{\mathbb{C}}\equiv  \left\langle \sigma_1,\sigma_2,\ldots ,\sigma_d \right\rangle_\mathbb{C}=\{v=v_1\sigma_1+v_2\sigma_2+\ldots+v_d\sigma_d:v_i\in \mathbb{C}\}$, with scalar multiplication and vector addition defined via 
\beqn
v+\lambda v'&=(v_1\sigma_1+v_2\sigma_2+\ldots+v_d\sigma_d)+\lambda(v'_1\sigma_1+v'_2\sigma_2+\ldots+v'_d\sigma_d)\\
& = (v_1+\lambda v'_1)\sigma_1+(v_2+\lambda v'_2)\sigma_2+\ldots +(v_d+\lambda v'_d)\sigma_d,\nonumber
\eqn 
for all $v,v'\in \langle G\rangle_{\mathbb{C}}$ and $\lambda\in \mathbb{C}$.
The \emph{regular representation,} $\rho_{\text{reg}}: G \rightarrow GL(d,\mathbb{C})$, is then defined by setting the group action 
\[\sigma :v\mapsto \sigma v=v_1(\sigma\sigma_1)+v_2(\sigma\sigma_2)+\ldots +v_d(\sigma\sigma_d),\]
for all $v\in \langle G\rangle_{\mathbb{C}}$ and $\sigma\in G$.
If we fix $\{\sigma_1,\sigma_2,\ldots ,\sigma_d\}$ as an ordered basis for $\langle G \rangle_{\mathbb{C}}$, it is then clear -- via Caley's theorem -- that each group element $\sigma$ gets mapped to a permutation matrix $K_\sigma:=\rho_{\text{reg}}(\sigma)$, with $K_\sigma\sigma_i=\sum_{j}\lrs{K_\sigma}^j_i\sigma_j:=\sigma\sigma_i$.
Thus $K_\sigma$ has matrix elements 
\beqn\label{eq:Ksigdef}
\left[K_\sigma\right]^{j}_{i}=\left\{\begin{array}{l}1, \text{ if } \sigma_j=\sigma\sigma_i, \\ 0, \text{ otherwise.}\end{array}\right.
\eqn
Consider the unit column vectors 
\[
\xi_1=(1,0,0,\ldots,0)^T,\quad \xi_2=(0,1,0,0,\ldots,0)^T,\quad\ldots\quad \xi_d=(0,0,\ldots,0,1)^T;
\] 
and identify $\sigma_i\equiv \xi_i$,  so that the group action becomes $\sigma:\xi_i\mapsto K_\sigma\xi_i=\xi_j$ where $\sigma_j=\sigma\sigma_i$.
Thus the matrix elements $\left[K_\sigma\right]^{j}_{i}$ have $i$ as the column label and $j$ as the row label.

With the regular representation in hand, it can then be shown (see \citet{sumner2011}) that the group-based model defined by $G$ has rate matrices of the form
\beqn
Q=-\lambda\id+\sum_{\epsilon\neq \sigma\in G} \alpha^{\sigma}K_\sigma,\nonumber
\eqn
where each $0\leq \alpha^\sigma \in \mathbb{R}$ and $\lambda=\sum_{\epsilon\neq \sigma\in G}\alpha^{\sigma}$.

The regular representation is one example of the general concept of a \emph{representation} of $G$ on a vector space $V$, defined as a homomorphism $\rho: G\rightarrow GL(V)$ satisfying $\rho(g_1g_2)=\rho(g_1)\rho(g_2)$ for all $g_1,g_2\in G$.
A representation is said to be \emph{reducible} if there exists a proper subspace $U\subset V$ satisfying $\rho(g)U\subset U$, i.e. the set of matrices $\rho(G)$ send vectors in $U$ back to $U$.
In this case, $U$ is called an \emph{invariant subspace}.
The representation $\rho$ is then called \emph{irreducible} if $V$ does not contain any invariant subspaces.

The reader should note that the usual construction of a ``group-based'' model  \citep{semple2003} stipulates that $G$ be \emph{abelian}. 
Although the construction just given using the regular representation allows for non-abelian $G$, we will nonetheless only consider the abelian case in this paper, because, as discussed in the introduction, it is only in the abelian case that a (linear) inversion of phylogenetic models is possible.
In this case the irreducible representations of $G$ are all one-dimensional \citep{sagan2001}, and hence reduce to the group characters, as is exploited in the previous approaches using Fourier analysis \citep{szekely1993}.

\subsection{Phylogenetic tensors}

As is shown in \citet{sumner2005} and in more detail in \citet{sumner2010}, phylogenetic distributions on the state space $\lrs{d}:=\{1,2,\ldots,d\}$ can be represented as tensors in the $n$-fold tensor product space $\otimes^n\mathbb{C}^{d}:=\mathbb{C}^{d}\otimes \mathbb{C}^{d} \otimes \ldots \otimes \mathbb{C}^{d}$. 
If we choose $\{\xi_1,\xi_2,\ldots,\xi_d\}$ as an ordered basis for $\mathbb{C}^{d}$, and ordered basis $\{\xi_{i_1}\otimes \xi_{i_2}\otimes \ldots \otimes \xi_{i_d}\}_{i_1,i_2,\ldots,i_n\in \lrs{d}}$ for the tensor product space, a ``phylogenetic tensor'' $P=\sum_{i_1,i_2,\ldots,i_n \in \lrs{d}}p_{i_1i_2\ldots i_n}\xi_{i_1}\otimes \xi_{i_2}\otimes \ldots \otimes \xi_{i_n}\in \otimes^n\mathbb{C}^{d}$ has the interpretation that the components $p_{i_1i_2\ldots i_n}$ represent the probability that the $n$ taxa take on the states $i_1,i_2,\ldots, i_n$ respectively.

Phylogenetic branching events can be generated by the linear operator $\delta:\mathbb{C}^d\rightarrow \mathbb{C}^{d}\otimes \mathbb{C}^{d}$ defined on the chosen basis via
\beqn
\delta(\xi_i):=\xi_i\otimes \xi_i,\qquad \delta(v)=\delta(\sum_i v_i\xi_i)=\sum_{i}v_i\delta(\xi_i)=\sum_iv_i\xi_i\otimes \xi_i.\nonumber
\eqn
The remarkable fact for group-based models, central to the present article, is that the rate matrices ``intertwine'' particularly simply with the branching operator:
\beqn
\delta(K_\sigma \xi_i)=\delta(\xi_{\sigma(i)})=\xi_{\sigma(i)}\otimes \xi_{\sigma(i)}=K_\sigma\otimes K_\sigma\cdot \delta(\xi_i).\nonumber
\eqn
Thus we have
\beqn
\delta\cdot Q=\left(-\lambda \id\otimes \id + \sum_{\epsilon\neq\sigma\in G }\alpha^\sigma K_{\sigma}\otimes K_{\sigma}\right)\cdot \delta,\nonumber
\eqn
which in turn implies (via the linearity of $\delta$) that 
\beqn\label{eq:intertwine}
\delta \cdot e^{Qt}=e^{-\lambda}\exp\left({\sum_{\epsilon\neq\sigma\in G }\alpha^\sigma K_{\sigma}\otimes K_{\sigma}}\right)\cdot \delta.
\eqn
This relation shows that mathematically, and hence conceptually, ``Markov evolution on a single followed by a branching event'' can be replaced with ``Branching event on a single taxa followed by (correlated) Markov evolution of two taxa.''
This equivalence is illustrated in Figure~\ref{fig:equiv}.

\begin{figure}[tbp]
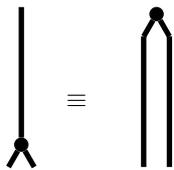

\centering

\begin{tabular}{ccc}

$
\psmatrix[colsep=.075cm,rowsep=.2cm]
&&&&[mnode=dot,dotscale=.00001]\\
\\
\\
\\
\\
\\
&&&&[mnode=dot,dotscale=1.5]\\
&&&[mnode=dot,dotscale=.00001]&&[mnode=dot,dotscale=.00001]\\
\ncline[linewidth=2pt]{1,5}{7,5}
\ncline[linewidth=2pt]{7,5}{8,4}
\ncline[linewidth=2pt]{7,5}{8,6}
\endpsmatrix
$ 

&

$
\psmatrix[colsep=1cm,rowsep=1cm]
\\
\equiv \\
\\
\endpsmatrix
$

&

$
\psmatrix[colsep=.075cm,rowsep=.2cm]
&&&&[mnode=dot,dotscale=1.5]\\
&&&[mnode=dot,dotscale=.00001]&&[mnode=dot,dotscale=.00001]\\
\\
\\
\\
\\
\\
&&&[mnode=dot,dotscale=.00001]&&[mnode=dot,dotscale=.00001]\\
\ncline[linewidth=2pt]{1,5}{2,4}
\ncline[linewidth=2pt]{1,5}{2,6}
\ncline[linewidth=2pt]{2,4}{8,4}
\ncline[linewidth=2pt]{2,6}{8,6}
\endpsmatrix
$ 
\end{tabular}
\caption{Markov evolution on a single followed by a branching event (illustrated on the left), is equivalent to a branching event on a single taxa followed by correlated Markov evolution of two taxa (illustrated on the right).
Mathematically, this equivalence can be implemented by exploiting the equality given in (\ref{eq:intertwine}).}
\label{fig:equiv}
\end{figure}

In \citet{sumner2010} we showed how to generalise this intertwining action to the case of the general Markov model.
Interestingly, the general intertwining has quite different structure from what occurs in group-based models, and the simplicity of (\ref{eq:intertwine}) is actually quite misleading for the general Markov model.
We refer the reader to \citet{sumner2010} for more discussion on this point.

Returning to the case of group-based models, for each subset $A\subseteq \lrs{n}$, we define a linear map on $\otimes^n\mathbb{C}^d$ as the tensor product $K^{(A)}_\sigma:=K_{\sigma}^{a_1}\otimes K_{\sigma}^{a_2} \otimes \ldots \otimes K_{\sigma}^{a_n}$ where $a_i=1$ if $i\in A$ and 0 otherwise.
For example, if $n=5$, we have
\beqn
K_{\sigma}^{(\{1,2,4\})}=K_\sigma \otimes K_{\sigma} \otimes \id \otimes K_{\sigma} \otimes \id.\nonumber
\eqn
To develop a phylogenetic tensor on a tree, we root the phylogenetic tree at taxon $n$, and label edges by subsets $\emptyset \neq e\subseteq\lrs{n-1}$, where $i\in e$ if the path from taxa $n$ to taxa $i$ crosses the edge labelled by $e$. 
A five taxa tree with this labelling, is presented in Figure~\ref{fig:tree}.
To each edge labelled by $\emptyset \neq e\subseteq\lrs{n-1}$, we assign the rate matrix
\beqn
Q_{e}:=-\lambda_e \id+\sum_{\epsilon\neq\sigma \in G}\alpha_{e}^\sigma K_\sigma,\nonumber
\eqn  
where each $\alpha_{e}^\sigma\geq 0$ is the rate of substitution for all states $\sigma_1$ to $\sigma_2$ satisfying $\sigma=\sigma_2\sigma_1^{-1}$, and $\lambda_e=\sum_{\sigma\in G}\alpha_{e}^\sigma$.
Each edge is then assigned substitution matrix $M_e=e^{Q_e}$, so that the time parameter for each edge is absorbed into the definition of $Q_e$. 

\begin{figure}[tbp]
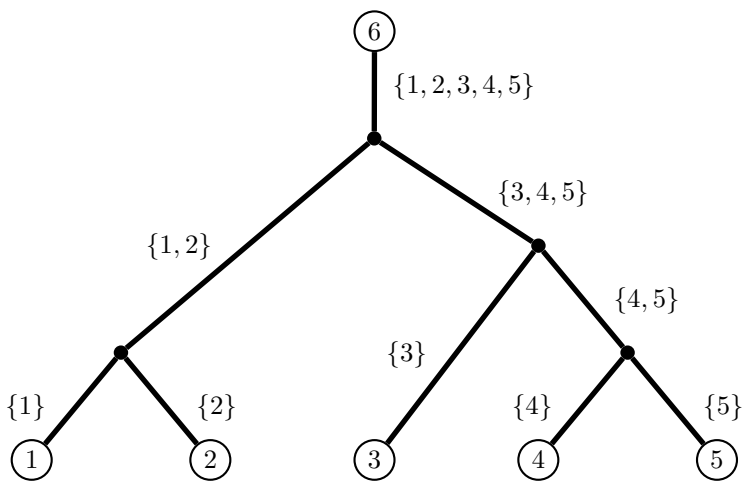

  \centering
$    
 \psmatrix[colsep=.8cm,rowsep=1cm,mnode=circle]
&&&& 6\\
&&&& [mnode=dot,dotscale=1.5]\\
&&&&&& [mnode=dot,dotscale=1.5]\\
& [mnode=dot,dotscale=1.5] &&&&&& [mnode=dot,dotscale=1.5] \\
1 & & 2 && 3 && 4 & & 5 
\ncline[linewidth=2pt]{1,5}{2,5}>{\{1,2,3,4,5\}}
\ncline[linewidth=2pt]{2,5}{3,7}>{\{3,4,5\}}
\ncline[linewidth=2pt]{3,7}{5,5}<{\{3\}}
\ncline[linewidth=2pt]{2,5}{4,2}<{\{1,2\}}
\ncline[linewidth=2pt]{4,2}{5,1}<{\{1\}}
\ncline[linewidth=2pt]{4,2}{5,3}>{\{2\}}
\ncline[linewidth=2pt]{3,7}{4,8}>{\{4,5\}}
\ncline[linewidth=2pt]{4,8}{5,7}<{\{4\}}
\ncline[linewidth=2pt]{4,8}{5,9}>{\{5\}}
\endpsmatrix
$

                    
  \caption{A six taxa tree rooted at taxon 6 with edges labelled by subsets of $\{1,2,3,4,5\}$.}
  \label{fig:tree}
\end{figure}

Now iterating (\ref{eq:intertwine}) multiple times, \citet{bashford2004,sumner2010} show that any phylogenetic tensor can be written as
\beqn\label{eq:networkrep}
P=e^{-\lambda}\exp\left(\sum_{\emptyset \neq e\subseteq \lrs{n-1},\sigma\in G} \alpha_{e}^\sigma K^{(e)}_\sigma\right)\cdot \delta^{n-1}\pi.
\eqn
where $\lambda=\sum_{\emptyset \neq e\subseteq \lrs{n-1}}\lambda_e=\sum_{\emptyset \neq e\subseteq \lrs{n-1},\epsilon\neq \sigma\in G} \alpha_{e}^\sigma$, and $\delta^{n-1}\pi$ is the $d\times d\times \ldots \times d$ tensor that represents the ``zero edge-length star tree'' distribution on $n$ taxa.
It is this form of phylogenetic tensors that will do a lot of the heavy lifting in the discussion that follows.
The reader should note that under this representation, there is no need for the edge parameters $\{\alpha_{e}^\sigma:\emptyset \neq e\subseteq \lrs{n-1},\sigma\in G\}$ to be chosen to be compatible with a particular tree, hence the possibilities for generalising to non-tree-like or network models, as discussed in the introduction.

The stationary distribution for group-based models is uniform (because the rate matrices are doubly stochastic).
In this paper we always assume a stationary distribution, so that:
\beqn
\pi=\fra{1}{d}(1,1,\ldots ,1)^T,\nonumber
\eqn
and $\delta^{n-1}\pi$ has tensor components
\beqn
\left[\delta^{n-1}\pi\right]_{i_1i_2\ldots i_n}=\left\{\begin{array}{ll}\frac{1}{d},\text{ if }i_1=i_2=\ldots=i_n,\\0,\text{ otherwise.}\end{array}\right.\nonumber
\eqn

This concludes our discussion of the tensor presentation of phylogenetic probability distributions under group-based models. 
We now review the standard Fourier analysis of these models, and make the connections to representation theory explicit.

\subsection{Connection to Fourier analysis}

In this subsection we briefly point out the connection between the standard Fourier analysis and representation theory. 
Understanding this connection -- or indeed the underlying representation theory -- is not required to understand our general method, so the uninterested reader may wish to skip forward directly to the next section.

Given an (abelian) group-based model the crucial aspects of the Fourier transform that are exploited in the phylogenetic context (e.g \cite{chor2000,evans1993,hendy1989,hendy1989b,hendy1993,hendy1994,hendy2008,sturmfels2004,szekely1993b,szekely1993}) are as follows.

\begin{res}\label{lem:fourier}
Let $f_1,f_2$ be functions from a finite abelian group $G$ to $\mathbb{C}$ and $\textbf{1}$ the constant function.
\begin{enumerate}
\item The group $G$ and the dual group $\widehat{G}:=\text{Hom}(G,\mathbb{C}^\times)$ are isomorphic as abstract groups.
\item Fourier transform turns convolution into multiplication, i.e., $\widehat{f_1\ast f_2}=\widehat{f_1}\cdot\widehat{f_2}$, and
\item If $\chi\in \widehat{G}$ is irreducible, then $\widehat{\textbf{1}}(\chi)=|G|$ if $\chi=1$ (the unit in $\widehat{G}$) and $\widehat{\textbf{1}}(\chi)=0$ otherwise.
\end{enumerate}
\end{res}

We recall the aspects of the representation theory that are needed in our discussion and express the above in terms of them.
For the reader who is unfamilar with the general theory, we recommend the excellent elementary text \citet{sagan2001}.

\begin{res}\label{thm:projectionops}
Given a representation $\rho:G\rightarrow GL(V)$ and an irreducible character $\chi: G \rightarrow \mathbb{C}$, the projectors onto the irreducible representations of $G$ are given by
\beqn
\Theta_{\chi}:=\fra{1}{|G|}\sum_{\sigma\in G}\chi(g)\rho(g).\nonumber
\eqn
\end{res}

\begin{res}
The regular representation contains each irreducible representation $\rho_\chi$ exactly $\dim(\rho_\chi)=\chi(e)$ times.
\end{res}

\begin{res}\label{thm:oneDreps}
The irreducible representations of an abelian group are one-dimensional.
\end{res}

\begin{res}
The character table of an abelian group $G$ diagonalizes the regular representation.
\end{res}

\begin{res}
Any (finitely generated) abelian group $G$ is isomorphic to a direct product of cyclic groups of prime-power order, ie. $G\cong \mathbb{Z}_{r_1}\times \mathbb{Z}_{r_2}\times \mathbb{Z}_{r_q}$ where each $r_i=q_i^{n_i}$ where $q_i$ is prime and $n_i$ is a positive integer.
\end{res}

\begin{res}
The irreducible representation of $\mathbb{Z}_{r}$ are given by $\rho_{i}(\sigma)=\omega^{i}$ with $i=0,1,2,\ldots, {r-1}$ and $\omega^{r}=1$.
\end{res}

\begin{res}
The irreducible representations of $\mathbb{Z}_{r_1}\times \mathbb{Z}_{r_2}\times\ldots\times \mathbb{Z}_{r_q}$ are given by $\rho_{i_1i_2\ldots i_q}=\rho^{(1)}_{i_1}\otimes \rho^{(2)}_{i_2}\otimes \ldots \otimes \rho^{(q)}_{i_q}$, where $\rho^{(i)}_k(\sigma_i)=(\omega_{i})^k$ and $(\sigma_i)^{r_i}=e_i$ with $e_i$ the identity in $\mathbb{Z}_{r_i}$ and $(\omega_{i})^{r_i}=1$.
\end{res}
\begin{proof}
The representation $\rho:=\rho_1\otimes\rho_2$ of $G=G_1\times G_2$, constructed from irreducible representations $\rho_1,\rho_2$ of $G_1,G_2$ respectively, is irreducible. 
The result follows from induction.
\end{proof}

\begin{res}
The Fourier analytic results of Result~\ref{lem:fourier} have representation theory counterparts:
\begin{enumerate}
\item The regular representation is faithful, i.e. injective.
\item The columns of the character table project onto the irreducible subspaces. Therefore, the character table of an abelian group $G$ diagonlizes the regular representation.
\item For an abelian group, the identity column of the character table obviously sums to $|G|$ and the other columns are orthogonal, thus the other columns sum to 0.
\end{enumerate}
\end{res}

In what follows, we discuss the inversion of abelian group-based models.
We present the simplest case with $G=\mathbb{Z}_2$ in \S\ref{sec:binsymm}; the $G=\mathbb{Z}_3$ case in \S\ref{sec:Z3}; the $G=\mathbb{Z}_2\times \mathbb{Z}_2$ case in \S\ref{sec:Z2Z2}; the general $G=\mathbb{Z}_r$ case in \S\ref{sec:Zr}; and finally we discuss the case of any abelian group in \S\ref{sec:anyabelian}.

\section{The binary-symmetric case}\label{sec:binsymm}
We begin with the inversion of the so-called ``binary-symmetric'' model.
Consider $\mathbb{C}^2$ with standard basis
\beqn
\left\{\xi_0=\left(\begin{array}{r}1 \\ 0\end{array}\right),\xi_1=\left(\begin{array}{r}0 \\ 1\end{array}\right)\right\}.\nonumber
\eqn
As a group-based model, the binary-symmetric model arises by taking the group 
\beqn
G:=\mathbb{Z}_2=\{0,1\}_{+ \text{(mod 2)}}\cong\lra{\sigma|\sigma^2=\epsilon},\nonumber
\eqn
with a generic rate matrix given by
\beqn
Q=\left(
\begin{array}{rr}
	-1 & 1 \\
	1  & -1 \\
\end{array}\right)=-\id+K,\nonumber
\eqn
where $K=\left(
\begin{array}{rr}
	0 & 1 \\
	1  & 0 \\
\end{array}\right)$ is the permutation matrix representing $\sigma$ in the standard basis.

Now
\[\rho_{\text{reg}}:\begin{array}{ll}\mathbb{Z}_2\rightarrow \mathbb{M}_{2}(\mathbb{C}) \\ \sigma \mapsto K\end{array}\] 
is the regular representation of $\mathbb{Z}_2$, and the character table of $\mathbb{Z}_2$ given in Table~\ref{tab:Z2chartab} is easily recognised to be the Hadamard matrix
\beqn
h=\left(
\begin{array}{rr}
	1 & 1 \\
	1  & -1 \\
\end{array}\right).\nonumber
\eqn
As $\mathbb{Z}_2$ is an abelian group, the irreducible representations are one-dimensional (Result~\ref{thm:oneDreps}).
Recalling Result~\ref{thm:projectionops}, the corresponding projection operators can be read off from the columns of the character table.
That is, the operators
\beqn
\Theta_{id}:&=\fra{1}{2}\left(\epsilon+\sigma\right),\nonumber\\
\Theta_{sgn}:&=\fra{1}{2}\left(\epsilon-\sigma\right);\nonumber
\eqn
project $\rho_{\text{reg}}=id\oplus sgn$ onto the $id$ and $sgn$ representations of $\mathbb{Z}_2$, respectively.

\begin{table}[t]
\centering
\begin{tabular}{|c|rr|}
\hline
 & $\texttt{id}$ & $\texttt{sgn}$  \\
\hline
$[e]$ & 1 & 1  \\
$[\sigma]$ & 1 & -1 \\
\hline
\end{tabular}
\caption{The character table of $\mathbb{Z}_2$.}
\label{tab:Z2chartab}
\end{table}

This observation prompts us to work in the alternative basis:
\beqn
f_0:&=\Theta_{id}\cdot \xi_0=\Theta_{id}\cdot \xi_1=h\xi_0=\xi_0+\xi_1,\nonumber\\
f_1:&=\Theta_{sgn}\cdot \xi_0=-\Theta_{sgn}\cdot \xi_1=h\xi_1=\xi_0-\xi_1.
\eqn
In this basis the permutation matrix is diagonal:
\beqn
\widehat{K}:&=hKh^{-1}=\left(
\begin{array}{rr}
	1 & 0 \\
	0  & -1 \\
\end{array}\right),\nonumber\\
\widehat{Q}:&=-\id+\widehat{K}=\left(
\begin{array}{rr}
	0 & 0 \\
	0  & -2 \\
\end{array}\right).
\eqn
The representation-theoretic perspective on $\widehat{K}$ is to observe that $id(\sigma)=1$ and $sgn(\sigma)=-1$.

Referring to (\ref{eq:networkrep}), we know that we can write a generic phylogenetic tensor as
\beqn
P=e^{-\lambda}\exp\left(\sum_{\emptyset \neq e\subseteq\lrs{n-1}}\alpha_eK^{(e)}\right)\cdot \delta^{n-1}\pi,\nonumber
\eqn
where $\lambda=\sum_{\emptyset \neq e\subseteq\lrs{n-1}}\alpha_e$.

We index matrix and tensor indices by using $i,j,k=0,1\in\mathbb{Z}_2$ and allow multiplication $\times$ in the ring of integers $\mathbb{Z}$.
The Hadamard matrix then has matrix elements $\left[h\right]^j_i=(-1)^{i\times j}$ where $j$ is the row index and $i$ is the column index.
Observe that in the diagonal basis, the permutation matrix has elements
\beqn
\left[\widehat{K}\right]_i^j=\delta_{ij}(-1)^i.\nonumber
\eqn
Thus we have expressions such as
\beqn
\lrs{\widehat{K}^{(\{2,3\})}}_{i_1i_2i_3}^{j_1j_2j_3}=\delta_{i_1j_1}\delta_{i_2j_2}\delta_{i_3j_3}(-1)^{i_2+i_3},\nonumber
\eqn
where $\widehat{K}^{(\{2,3\})}=\id \otimes \widehat{K}\otimes \widehat{K}$.

As we are dealing with tensors of arbitrary size, it is convenient to represent a string such as $i_1i_2\ldots i_n$ as an \emph{ordered-bipartition} $\mu=\mu_0\sep\mu_1$ of the set $\lrs{n}$, where $\mu_0,\mu_1\subseteq \lrs{n}$ with $j\in \mu_k$ if and only if $i_j=k$.
For example we have the following equivalences:
\beqn
00110\equiv \{1,2,5\}\sep\{3,4\},\qquad 01111\equiv \{1\}\sep\{2,3,4,5\},\qquad 10001\equiv \{2,3,4\}\sep\{1,5\},\nonumber
\eqn
and inequivalence:
\beqn
01010\equiv \{1,3,5\}\sep\{2,4\}\neq \{2,4\}\sep\{1,3,5\}\equiv 10101.\nonumber
\eqn
We then have 
\beqn
\lrs{\widehat{K}^{(e)}}_{i_1i_2\ldots i_n}^{j_1j_2\ldots j_n}=\lrs{\widehat{K}^{(e)}}_{\mu}^{\nu}=\lrs{\widehat{K}^{(e)}}_{\mu_0\sep\mu_1}^{\nu_0\sep\nu_1}=\delta_{\mu_0\nu_0}\delta_{\mu_1\nu_1}(-1)^{|e\cap \mu_1|}.\nonumber
\eqn

Defining $h^{(n)}:=h^{(n-1)}\otimes h$ where $h^{(1)}:=h$, in the diagonal basis $\widehat{P}:=h^{(n)}\cdot P$ and using our notation $h^{(n)}$ has tensor components
\beqn
\lrs{h^{(n)}}_{\mu}^{\nu}=\lrs{h^{(n)}}_{\mu_0\sep\mu_1}^{\nu_0\sep\nu_1}=\lrs{h^{(n)}}_{i_1i_2\ldots i_n}^{j_1j_2\ldots j_n}=(-1)^{i_1\times j_1+i_2\times j_2+\ldots +i_n\times j_n}=(-1)^{|\mu_1\cap \nu_1|}.\nonumber
\eqn
The zero edge-length star-tree initial distribution has tensor components
\beqn
\lrs{\delta^{n-1}\pi}_{i_1i_2\ldots i_n}=\fra{1}{2}\delta_{i_1i_2}\delta_{i_1i_3}\ldots \delta_{i_1i_n},\nonumber
\eqn
(where, although it seems we have given preference to the taxa 1 in this expression, there are many ways that this distribution can be expressed using the $\delta_{ij}$).
In the diagonal basis with $\widehat{\delta^{n-1}\pi}:=h^{(n)}\cdot \delta^{n-1}\pi$,  we have components
\beqn
\lrs{\widehat{\delta^{n-1}\pi}}_{i_1i_2\ldots i_n}&=\fra{1}{2}\sum_{j_1,j_2,\ldots,j_n}(-1)^{i_1\times j_1+i_2\times j_2+\ldots +i_n\times j_n}\delta_{j_1j_2}\delta_{j_1j_3}\ldots \delta_{j_1j_n}\\
&=\fra{1}{2}\sum_{j_1}(-1)^{\left(i_1+i_2+\ldots +i_n\right)\times j_1}\\
&=\fra{1}{2}\left(1+(-1)^{i_1+i_2+\ldots +i_n}\right),\nonumber
\eqn
which is exactly the statement
\beqn
\lrs{\widehat{\delta^{n-1}\pi}}_{\mu}=\lrs{\widehat{\delta^{n-1}\pi}}_{\mu_0\sep\mu_1}=\fra{1}{2}\left(1+(-1)^{|\mu_1|}\right).\nonumber
\eqn

Since $\widehat{K}$ is diagonal in the transformed basis, we can conclude that
\beqn
\lrs{\widehat{P}}_{\mu}=\lrs{\widehat{P}}_{\mu_0\sep\mu_1}=e^{-\lambda}\exp\left(\sum_{\emptyset \neq e\subseteq \lrs{2,n}}\alpha_e\lrs{\widehat{K}^{(e)}}^{\mu_0\sep\mu_1}_{\mu_0\sep\mu_1}\right)\fra{1}{2}\left(1+(-1)^{|\mu_1|}\right).\nonumber
\eqn
Of course many of these tensor components will be zero and we would like to ignore these.

Take $u=u_0\sep u_1$ as an ordered bipartition of the reduced set $\lrs{n-1}$, so that $u\equiv i_1i_2\ldots i_{n-1}$ where $j\in u_k$  if and only if $i_j=k$, and define
\beqn
\gamma(u)&=\left\{\begin{array}{ll} 0, \text{ if }|u_1|\text{ is even,}\\1,\text{ if }|u_1|\text{ is odd;}\end{array}\right.\\
&=2-\left(0|u_0|+1|u_1|\right)\text{ (mod 2)},\nonumber
\eqn
and interpret $u\cdot \gamma(u)$ as a string: $u\cdot \gamma(u)= i_1i_2\ldots i_{n-1}\gamma(u)$.

If we make the definitions
\beqn
\mathcal{P}_{u}:=\lrs{\widehat{P}}_{u\cdot \gamma(u)},\qquad \eta_{u}:=\fra{1}{2}\sum_{\emptyset \neq e\subseteq \lrs{n-1}}\alpha_{e}\lrs{\widehat{K}^{(e)}}^{u\cdot \gamma(u)}_{u\cdot \gamma(u)},\nonumber
\eqn
then we can write the non-zero components as
\beqn
\mathcal{P}_u=e^{-\lambda}\exp\left(\eta_{u}\right),\nonumber
\eqn
with inverses
\beqn\label{eq:binsymm1stinv}
\eta_{u}=\ln\left(\mathcal{P}_u\right)+\lambda.
\eqn
This is the first part of the inversion.

We would like to go further and actually recover the individual edge weights $\alpha_e$.
To do this we define the (square) $2^{n-1}\times 2^{n-1}$ matrix $F$ with components
\beqn
\lrs{F}_u^e:=\lrs{\widehat{K}^{(e)}}^{u\cdot \gamma(u)}_{u\cdot \gamma(u)}=(-1)^{|e\cap u|}=\lrs{h^{(n-1)}}^{e}_u,\nonumber
\eqn
with $e$ a subset and $u$ an ordered-bipartition of $\lrs{n-1}$.
As $(h^{(n-1)})^2=\fra{1}{2^{n-1}}\id$, we see that $F$ provides its own inverse $F^{-1}$ with components
\beqn
\lrs{F^{-1}}_e^u:=\fra{1}{2^{n-1}}\lrs{F}_{u}^e.\nonumber
\eqn

Defining the column vectors $\vec{\alpha}=\left\{\alpha_e\right\}$ and $\vec{\eta}=\left\{\eta_{u}\right\}$, we can write the matrix equations
\beqn
\vec{\eta}=F\vec{\alpha},\qquad \vec{\alpha}=F^{-1}\vec{\eta}.\nonumber
\eqn
Together with the first part of the inversion (\ref{eq:binsymm1stinv}), these equations give a one-one map between pattern probabilities and edge weights for the binary-symmetric model.

\section{Inversion of the $\mathbb{Z}_3$ model}\label{sec:Z3}

Taking confidence from the previous case we now discuss the inversion of the group-based phylogenetic model with $G=\mathbb{Z}_3$.
We take $\mathbb{Z}_3=\{0,1,2\}_{+\text{ (mod 3)}}\cong\lra{\sigma|\sigma^3=\epsilon}$ and, by analogy to the $\mathbb{Z}_2$ case, index tensors with indices $i,j=0,1,2$ and allow multiplication $\times$ by extending $\mathbb{Z}_3$ to the ring $\mathbb{F}_3=\{0,1,2\}_{+,\times \text{ (mod 3)}}$.

In this case a generic rate matrix is given by
\beqn
Q&=\left(\begin{array}{ccc}
-(\alpha+\beta) & \beta & \alpha \\
\alpha & -(\alpha+\beta) & \beta \\
\beta & \alpha & -(\alpha+\beta) \\
\end{array}\right)\\
&=-\left(\alpha+\beta\right)\id+\alpha K_1+\beta K_2,\nonumber
\eqn
where 
\beqn
K_1=\left(\begin{array}{ccc}
0 & 0 & 1 \\
1 & 0 & 0 \\
0 & 1 & 0 \\
\end{array}\right),\qquad
K_2=\left(\begin{array}{ccc}
0 & 1 & 0 \\
0 & 0 & 1 \\
1 & 0 & 0 \\
\end{array}\right),\nonumber
\eqn
are the matrices representing the permutations $\sigma\cong (123)$ and $\sigma^2\cong (132)$ under the regular representation, respectively.

We define $\omega=e^{{2\pi i}/{3}}$, and present the character table of $\mathbb{Z}_3$ is given in Table~\ref{tab:Z3chartab}.
The decomposition of the regular representation is $\rho_{\text{reg}}=id\oplus\omega\oplus\omega^2$, and the columns of the character table give the projection operators onto the (one-dimensional) irreducible subspaces:
\beqn
\Theta_{id}:&=\fra{1}{3}\left(\epsilon+\sigma+\sigma^2\right)\\\nonumber
\Theta_{\omega}:&=\fra{1}{3}\left(\epsilon+\omega\sigma+\omega^2\sigma^2\right)\\\nonumber
\Theta_{\omega^2}:&=\fra{1}{3}\left(\epsilon+\omega^2\sigma+\omega\sigma^2\right)\\\nonumber
\eqn
Therefore, the matrix
\beqn
f=\left(\begin{array}{ccc}
1 & 1 & 1 \\
1 & \omega & \omega^2 \\
1 & \omega^2 & \omega \\
\end{array}\right),\nonumber
\eqn
diagonalizes the generic rate matrix for this model:
\beqn
\widehat{Q}=fQf^{-1}=
\left(\begin{array}{ccc}
0 & 0 & 0 \\
0 & \alpha\omega+\beta\omega^2 & 0 \\
0 & 0 & \alpha\omega^2+\beta\omega \\
\end{array}\right),\nonumber
\eqn
or, equivalently,
\beqn
\widehat{K}_1=fK_1 f^{-1}=
\left(\begin{array}{ccc}
1 & 0 & 0 \\
0 & \omega & 0 \\
0 & 0 & \omega^2 \\
\end{array}\right),\qquad
\widehat{K}_2=fK_2 f^{-1}=
\left(\begin{array}{ccc}
1 & 0 & 0 \\
0 & \omega^2 & 0 \\
0 & 0 & \omega \\
\end{array}\right).\nonumber
\eqn

\begin{table}[t]
\centering
\begin{tabular}{|c|ccc|}
\hline
 & $\texttt{id}$ & $\omega$ & $\omega^2$  \\
\hline
$[e]$ & 1 & 1 & 1  \\
$[\sigma]$ & 1 & $\omega$ & $\omega^2$ \\
$[\sigma^2]$ & 1 & $\omega^2$ & $\omega$ \\
\hline
\end{tabular}
\caption{The character table of $\mathbb{Z}_3$.}
\label{tab:Z3chartab}
\end{table}

We recall our basic result (\ref{eq:networkrep}) that for group-based models, a generic phylogenetic tensor can be expressed as
\beqn
P=e^{-\lambda}\exp\left(\sum_{\emptyset \neq e\subseteq\lrs{n-1}} \left(\alpha_{e}K^{(e)}_1+\beta_{e}K^{(e)}_2\right)\right)\cdot \delta^{n-1}\pi,\nonumber
\eqn
where $\lambda=\sum_{\emptyset \neq e\subseteq \lrs{n-1}}\left(\alpha_e+\beta_e\right)$.
We take the stationary distribution as initial distribution, so $\pi=(\fra{1}{3},\fra{1}{3},\fra{1}{3})^T$.

The matrix elements of $f$ can be expressed as $\left[f\right]_i^j=\omega^{i\times j}$, where we extend $i,j\in\mathbb{Z}_3$ to include multiplication $\times$ from the ring of integers $\mathbb{Z}$.
Similarly,
\beqn
\lrs{\widehat{K}_1}_i^j=\delta_{ij}\omega^i,\qquad \lrs{\widehat{K}_2}_i^j=\delta_{ij}(\omega^2)^i.\nonumber
\eqn
More generally, tensorial components can be expressed as
\beqn
\lrs{\id\otimes \widehat{K_1}\otimes \widehat{K_1} }_{i_1i_2i_3}^{j_1j_2j_3}=\delta_{i_1j_1}\delta_{i_2j_2}\delta_{i_3j_3}\omega^{i_2+i_3}.\nonumber
\eqn

We represent a string $i_1i_2\ldots i_n$ as an \emph{ordered-tripartition}, $i_1i_2\ldots i_n\equiv \mu=\mu_0\sep\mu_1\sep\mu_2$, of the set $\lrs{n}$, where $j\in \mu_k$ if and only if $i_j=k$.
For example, if we take $n=5$, we have
\beqn
00000&\equiv \{1,2,3,4,5\}\sep\emptyset \sep\emptyset, \nonumber\\
00120&\equiv \{1,2,5\}\sep \{3\}\sep \{4\},\\
01122&\equiv \{1\}\sep \{2,3\}\sep \{4,5\}.
\eqn
Taking $n\!=\!3$, we have
\beqn
\lrs{\widehat{K}_1^{(\{2,3\})}}_\mu^\nu=
\lrs{\id\otimes \widehat{K}_1\otimes \widehat{K}_1}_\mu^\nu=\lrs{\id\otimes \widehat{K}_1\otimes \widehat{K}_1}_{\mu_0\sep \mu_1\sep \mu_2}^{\nu_0\sep \nu_1\sep \nu_2}=\delta_{\mu\nu}\omega^{|\mu_1\cap \{2,3\}|+2|\mu_2\cap \{2,3\}|},\nonumber
\eqn
and in general:
\beqn
\lrs{\widehat{K}^{(e)}_1}_\mu^\nu=\delta_{\mu\nu}\omega^{|e\cap \mu_1|+2|e\cap \mu_2|},\qquad \lrs{\widehat{K}^{(e)}_2}_\mu^\nu=\delta_{\mu\nu}\omega^{|e\cap \mu_2|+2|e\cap \mu_1|}.\nonumber
\eqn

Taking the uniform distribution as initial distribution, the intial star-tree distribution can be written as
\beqn
\lrs{\delta^{n-1}\pi}_{i_1i_2\ldots i_n}=\fra{1}{3}\delta_{i_1i_2}\delta_{i_1i_3}\ldots \delta_{i_1i_n}.\nonumber
\eqn
Defining $f^{(n)}=f^{(n-1)}\otimes f$ where $f^{(1)}=f$, we have
\beqn
\lrs{f^{(n)}}_\mu^\nu=\lrs{f^{(n)}}_{i_1i_2\ldots i_n}^{j_1j_2\ldots j_n}=\lrs{f}_{i_1}^{j_1}\lrs{f}_{i_2}^{j_2}\ldots\lrs{f}_{i_n}^{j_n} =\omega^{i_1\times j_1+i_2\times j_2+\ldots +i_n\times j_n},\nonumber
\eqn
and in the transformed basis, where $\widehat{\delta^{n-1}\pi}:=f^{(n)}\cdot \delta^{n-1}\pi$, we have
\beqn
\lrs{\widehat{\delta^{n-1}\pi}}_{i_1i_2\ldots i_n}&=\fra{1}{3}\sum_{j_1,j_2,\ldots , j_n}\omega^{i_1\times j_1+i_2\times j_2+\ldots +i_n\times j_n}\delta_{j_1j_2}\delta_{j_1j_3}\ldots\delta_{j_1j_n}\nonumber\\
&=\fra{1}{3}\sum_{j_1}\omega^{j_1\times (i_1+i_2+\ldots+i_n)}\\
&=\fra{1}{3}\left(1+\omega^{i_1+i_2+\ldots +i_n}+(\omega^2)^{i_1+i_2+\ldots +i_n}\right).
\eqn
Indexing by ordered-tripartitions, we conclude that
\beqn
\lrs{\widehat{\delta^{n-1}\pi}}_\mu&=\fra{1}{3}\left(1+\omega^{i_1+i_2+\ldots +i_n}+(\omega^2)^{i_1+i_2+\ldots +i_n}\right)\\
&=\fra{1}{3}\left(1+\omega^{|\mu_1|+2|\mu_2|}+(\omega^2)^{|\mu_1|+2|\mu_2|}\right).\nonumber
\eqn

Now suppose $|\mu_1|+2|\mu_2|=0$ (mod 3), then
\beqn
\lrs{\widehat{\delta^{n-1}\pi}}_\mu=\fra{1}{3}\left(1+1+1\right)=1.\nonumber
\eqn
If $|\mu_1|+2|\mu_2|=1$ (mod 3), then 
\beqn
\lrs{\widehat{\delta^{n-1}\pi}}_\mu=\fra{1}{3}\left(1+\omega+\omega^2\right)=0,\nonumber
\eqn
and if $|\mu_1|+2|\mu_2|=2$ (mod 3), then
\beqn
\lrs{\widehat{\delta^{n-1}\pi}}_\mu=\fra{1}{3}\left(1+\omega^2+\omega\right)=0.\nonumber
\eqn
Thus we have found a basis where all the elements of the initial star-tree tensor are zero \emph{unless} the tripartion $\mu$ satisfies $|\mu_1|+2|\mu_2|=0$ (mod 3).
Crucially, this statement also holds for the phylogenetic tensor $\widehat{P}$ because in this basis the rate matrices of this model are diagonal:
\beqn
\lrs{\widehat{P}}_{\mu}&=\lrs{\widehat{P}}_{\mu_0\sep \mu_1\sep \mu_2}\\
&=e^{-\lambda}\exp\left(\fra{1}{2}\sum_{\emptyset \neq e \subseteq \lrs{n-1}}\lrs{\alpha_e K^{(e)}_1+\beta_{e}K^{(e)}_2 }^{\mu_0\sep \mu_1\sep \mu_2}_{\mu_0\sep \mu_1\sep \mu_2}\right)\fra{1}{3}\left(1+\omega^{1|\mu_1|}+\omega^{2|\mu_2|}\right).\nonumber
\eqn

%

We deal with this condition on $\mu$ by taking $u=u_0\sep u_1\sep u_2$ as an ordered-tripartion of the reduced set $\lrs{n-1}$ and setting $\mu=u\cdot \gamma(u)$ (considered as the concatenation of strings) where 
\beqn
\gamma(u)&=\left\{\begin{array}{ll} 0, \text{ if }|u_1|+2|u_2|=0\\1, \text{ if }|u_1|+2|u_2|=2 \\ 2; \text{ if }|u_1|+2|u_2|=1\end{array}\right.\nonumber\\
&=3-\left(0|u_0|+1|u_1|+2|u_2|\right)\text{ (mod 3)}.
\eqn
If we make the definitions
\beqn
\mathcal{P}_u:=\lrs{\widehat{P}}_{u\cdot \gamma(u)},\qquad\eta_{u}:=\lrs{\sum_{\emptyset \neq e\subseteq\lrs{n-1}}\alpha_eK_1^{(e)}+\beta_eK^{(e)}_2}^{u\cdot \gamma(u)}_{u\cdot \gamma(u)},\nonumber\\
\eqn
we then have the first part of the inversion
\beqn\label{eq:1stpartZ3}
\mathcal{P}_u=e^{-\lambda}\exp\left(\eta_u\right),\qquad \eta_u=\ln\left(\mathcal{P}_u\right)+\lambda.
\eqn

As in the $\mathbb{Z}_2$ case, we would like to use $\eta_u$ to recover the rate parameters $\alpha_e,\beta_e$ for all $\emptyset \neq e\subseteq \lrs{n-1}$ and thus complete the full inversion for this model.
Of course, it is little bit more difficult this time.

Recall that $\mu=\mu_0\sep \mu_1\sep \mu_2$ with $\mu_i\subseteq \lrs{n}$, whereas $u=u_0\sep u_1\sep u_2$ with $u_i\subseteq \lrs{n-1}$, and $\emptyset \neq e\subseteq \lrs{n-1}$.
Considering
\beqn
\lrs{K^{(e)}_1}^{\mu}_{\mu}&=\omega^{|e\cap \mu_1|+2|e\cap \mu_2|},\nonumber\\
\eqn
it follows that
\beqn
\lrs{K^{(e)}_1}^{u\cdot \gamma(u)}_{u\cdot \gamma(u)}&=\omega^{|e\cap u_1|+2|e\cap u_2|},\nonumber
\eqn
and similarly
\beqn
\lrs{K^{(e)}_2}^{u\cdot \gamma(u)}_{u\cdot \gamma(u)}=\omega^{|e\cap u_2|+2|e\cap u_1|}.\nonumber
\eqn
We make the observation that
\beqn
\lrs{F_1}_u^e:=\lrs{f^{(n-1)}}_{u_0\sep u_1\sep u_2}^{{e}^c\sep e\sep \emptyset}=\omega^{|u_1\cap e|+2|u_2\cap e|}=\lrs{K^{(e)}_\alpha}^{u\cdot \gamma(u)}_{u\cdot \gamma(u)},\nonumber
\eqn
and
\beqn
\lrs{F_2}_u^e:=\lrs{f^{(n-1)}}_{u_0\sep u_1\sep u_2}^{e^c\sep \emptyset\sep e}=\omega^{|u_2\cap e|+2|u_1\cap e|}=\lrs{K^{(e)}_\beta}^{u\cdot \gamma(u)}_{u\cdot \gamma(u)},\nonumber
\eqn
where $F_1$ and $F_2$ are $2^{n-1}\times 3^{n-1}$ matrices.

Thus we may write
\beqn
\eta_u=\sum_{\emptyset \neq e\subseteq \lrs{n-1}}\alpha_e \lrs{F_1}_{u}^e+\beta_{e}\lrs{F_2}_{u}^e.\nonumber
\eqn
Defining the column vectors $\vec{\alpha}=\{\alpha_e\},\vec{\beta}=\{\beta_e\}$ and $\vec{\eta}=\{\eta_{u}\}$, we can write
\beqn
\vec{\eta}=F_1\vec{\alpha}+F_2\vec{\beta},\nonumber
\eqn
and define two $3^{n-1}\times 2^{n-1}$ matrices $G_1$ and $G_2$ as
\beqn
\lrs{G_1}_e^u:&=\lrs{{f^{-1}}^{(n-1)}}_{{e}^c\sep e\sep \emptyset}^{u},\nonumber\\
\lrs{G_2}_e^u:&=\lrs{{f^{-1}}^{(n-1)}}_{e^c\sep \emptyset\sep e}^{u},
\eqn
where
\beqn
f^{-1}=\left(\begin{array}{ccc}
1 & 1 & 1 \\
1 & \omega & \omega^2 \\
1 & \omega^2 & \omega \\
\end{array}\right),\nonumber
\eqn
with $f f^{-1}=\id$.

Considering that
\beqn
\sum_{v}\lrs{{f^{-1}}^{(n-1)}}_{u}^v\lrs{f^{(n-1)}}_{v}^{w}=\delta_{uw},\nonumber
\eqn
for all ordered-triparitions $u,w$ of $\lrs{n-1}$, we have the matrix products
\beqn
\begin{array}{cc}
G_1  F_1=\id, & G_1 F_2 =0, \\
G_2  F_2=\id, & G_2  F_1 =0.\nonumber
\end{array}
\eqn
Thus the second part of the inversion for this model is
\beqn
\vec{\alpha}=G_1\vec{\eta},\qquad \vec{\beta}=G_2\vec{\eta}.\nonumber
\eqn
Together with (\ref{eq:1stpartZ3}), these equations give a one-one map between pattern probabilities and edge weights for the group-based model with $G=\mathbb{Z}_3$.

\section{Inversion of the K3ST model}\label{sec:Z2Z2}

We now consider the K3ST model \citep{kimura1981} which occurs as the group-based model with $G=\mathbb{Z}_2\times \mathbb{Z}_2=\{(0,0),(0,1),(1,0),(1,1)\}_{+\text{ (mod 2)}}\cong \lra{(12)(34),(13)(24)}$.
In this model a generic rate matrix is given by
\beqn
Q=-\left(\alpha+\beta+\gamma\right)\id+\alpha K_{01}+\beta K_{10}+\gamma K_{11},\nonumber
\eqn
where 
\beqn\label{eq:k3stperm}
K_{01}&=\id\otimes K=\left(\begin{array}{cccc}
0 & 1 & 0 & 0 \\
1 & 0 & 0 & 0 \\
0 & 0 & 0 & 1 \\
0 & 0 & 1 & 0 
\end{array}\right),\quad
K_{10}=K\otimes \id=\left(\begin{array}{cccc}
0 & 0 & 1 & 0 \\
0 & 0 & 0 & 1 \\
1 & 0 & 0 & 0 \\
0 & 1 & 0 & 0 
\end{array}\right),\\
K_{11}&=K\otimes K=\left(\begin{array}{cccc}
0 & 0 & 0 & 1 \\
0 & 0 & 1 & 0 \\
0 & 1 & 0 & 0 \\
1 & 0 & 0 & 0 
\end{array}\right).
\eqn
We already know that the $2\times 2$ Hadamard matrix $h$ diagonalizes $K$, so we see immediately that $H=h\otimes h$ diagonalizes this model:
\beqn
\widehat{K}_{01}:&=HK_{01} H^{-1}=h\otimes h \cdot \id\otimes K \cdot h^{-1}\otimes h^{-1}=\id\otimes hKh^{-1}=
\left(\begin{array}{cccc}
1 & 0 & 0 & 0 \\
0 & -1 & 0 & 0 \\
0 & 0 & 1 & 0 \\
0 & 0 & 0 & -1 
\end{array}\right),\\
\widehat{K}_{10}:&=HK_{10} H^{-1}=\left(\begin{array}{cccc}
1 & 0 & 0 & 0 \\
0 & 1 & 0 & 0 \\
0 & 0 & -1 & 0 \\
0 & 0 & 0 & -1 
\end{array}\right),\qquad
\widehat{K}_{11}:=HK_{11} H^{-1}=\left(\begin{array}{cccc}
1 & 0 & 0 & 0 \\
0 & -1 & 0 & 0 \\
0 & 0 & -1 & 0 \\
0 & 0 & 0 & 1 
\end{array}\right).\nonumber
\eqn

Of course $H$ is the character table of $\mathbb{Z}_2\times \mathbb{Z}_2$ and the permutation matrices (\ref{eq:k3stperm}), together with $K_{00}:=\id$, give the regular representation $\rho_{\text{reg}}\cong id\otimes id\oplus id\otimes sgn\oplus sgn\otimes id\oplus sgn\otimes sgn$, where we recall the basic result that the tensor product of two irreducible representations of a group $G$ gives an irreducible representation of $G\times G$.

Simplifying notation, for this model we index tensors with indices given as pairs: $i,j=00,01,10,11\in \mathbb{Z}_2\times \mathbb{Z}_2$; and we express the individual parts using lower case Roman characters.
For example, we write $i:=ab=01$, with $a=0$ and $b=1$.
This gives matrix elements:
\beqn
\lrs{\widehat{K}_{01}}_{ab}^{cd}&=\delta_{ac}\delta_{bd}(-1)^b,\nonumber\\
\lrs{\widehat{K}_{10}}_{ab}^{cd}&=\delta_{ac}\delta_{bd}(-1)^a,\\
\lrs{\widehat{K}_{11}}_{ab}^{cd}&=\delta_{ac}\delta_{bd}(-1)^{a+b},
\eqn
and more complicated tensor products such as
\beqn
\lrs{\widehat{K}_{01}\otimes \widehat{K}_{01}\otimes \id}_{a_1b_1a_2b_2 a_3b_3}^{c_1d_1c_2d_2c_3d_3}=\delta_{a_1c_1}\delta_{b_1d_1}\delta_{a_2c_2}\delta_{b_2d_2}\delta_{a_3c_3}\delta_{b_3d_3}(-1)^{b_1+b_2}.\nonumber
\eqn
Again we interpret strings such as $\mu\equiv a_1a_2\ldots a_n$ and $\nu\equiv b_1b_2\ldots b_n$ as ordered-bipartitions $\mu=\mu_0\sep \mu_1$ and $\nu=\nu_0\sep \nu_1$ of the set $\lrs{n}$.
We can then write matrix elements of tensor products as
\beqn
\lrs{\widehat{K}^{(e)}_{01}}_{\mu,\nu}^{\mu',\nu'}&=\delta_{\mu\mu'}\delta_{\nu\nu'}(-1)^{|e\cap \nu_1|},\\\nonumber
\lrs{\widehat{K}^{(e)}_{10}}_{\mu,\nu}^{\mu',\nu'}&=\delta_{\mu\mu'}\delta_{\nu\nu'}(-1)^{|e\cap \mu_1|},\\
\lrs{\widehat{K}^{(e)}_{11}}_{\mu,\nu}^{\mu',\nu'}&=\delta_{\mu\mu'}\delta_{\nu\nu'}(-1)^{|e\cap \mu_1|+|e\cap \nu_1|}.
\eqn

Taking the stationary distribution $\pi=\frac{1}{4}(1,1,1,1)^T$ as initial distribution, the zero edge-length star-tree distribution is given by
\beqn
\lrs{\delta^{n-1}\pi}_{i_1i_2\ldots i_n}=\fra{1}{4}\delta_{i_1i_2}\delta_{i_1i_3}\ldots \delta_{i_1i_n},\nonumber
\eqn
which in the finer index representation is
\beqn
\lrs{\delta^{n-1}\pi}_{a_1b_1a_2b_2\ldots a_nb_n}=\fra{1}{4}\delta_{a_1a_2}\delta_{a_1a_3}\ldots \delta_{a_1a_n}\delta_{b_1b_2}\delta_{b_1b_3}\ldots \delta_{b_1b_n}.\nonumber
\eqn

Recall that elements of the Hadamard matrix can be written as $\lrs{h}^a_b=(-1)^{a\times b}$, where $a,b\in\mathbb{Z}_2$ and we allow multiplication $\times$ by extending to the ring of integers $\mathbb{Z}$.
In the transformed basis, we have
\beqn
\lrs{\widehat{{\delta^{n-1}\pi}}}_{a_1b_1a_2b_2\ldots a_nb_n}&=\lrs{\widehat{\delta^{n-1}\pi}}_{\mu,\nu}\\
&\hspace{-8em}=\fra{1}{4}\sum_{c_1,c_2,\ldots,c_n}^{d_1,d_2,\ldots ,d_n}\lrs{h}^{a_1}_{c_1}\lrs{h}^{a_2}_{c_2}\ldots \lrs{h}^{a_n}_{c_n}\lrs{h}^{b_1}_{d_1}\lrs{h}^{b_2}_{d_2}\ldots \lrs{h}^{b_n}_{d_n}\delta_{a_1a_2}\delta_{a_1a_3}\ldots \delta_{a_1a_n}\delta_{b_1b_2}\delta_{b_1b_3}\ldots \delta_{b_1b_n}\\
&\hspace{-8em}=\fra{1}{4}\sum_{c_1,d_1}(-1)^{(a_1+a_2+\ldots a_n)\times c_1+(b_1+b_2+\ldots +b_n)\times d_1}\\
&\hspace{-8em}=\fra{1}{4}\left(1+(-1)^{a_1+a_2+\ldots +a_n}+(-1)^{b_1+b_2+\ldots +b_n}+(-1)^{a_1+a_2+\ldots +a_n+b_1+b_2+\ldots +b_n}\right)\\
&\hspace{-8em}=\left\{\begin{array}{ll} 0, \text{ if either }|\mu_1|\text{ or }|\nu_1|\text{ is odd};\\1, \text{ if }|\mu_1|\text{ and }|\nu_1|\text{ are both even}.\end{array}\right.\nonumber
\eqn

We recall (\ref{eq:networkrep}), so under this model we can express a generic phylogenetic tensor as
\beqn
P=e^{-\lambda}\exp\left(\sum_{\emptyset \neq e\subseteq \lrs{n-1}}\alpha_{e}K^{(e)}_{01}+\beta_{e}K^{(e)}_{10}+\gamma_{e}K^{(e)}_{11}\right) \cdot \delta^{n-1}\pi.\nonumber
\eqn 

To exclude the vanishing components we define, for all ordered bipartitions $u=u_0\sep u_1$ of the reduced set $\lrs{n-1}$,
\beqn
\gamma(u)&=\left\{\begin{array}{ll} 0, \text{ if }|u_1|\text{ is even},\\1, \text{ if }|u_1|\text{ is odd};\end{array}\right.\\
&=2-(0|u_0|+1|u_1|) \text{ (mod 2)},\nonumber
\eqn
and intepret $u\cdot \gamma(u)$ as the string $u\cdot \gamma(u)=a_1a_2\ldots a_{n-1}\gamma(u)$.
Then, for each pair $u,v$ of ordered-bipartitions of $\lrs{n-1}$, we define 
\beqn
\eta_{u,v}:=\lrs{\sum_{\emptyset \neq e\subseteq \lrs{n-1}}\alpha_{e}K^{(e)}_{01}+\beta_{e}K^{(e)}_{10}+\gamma_{e}K^{(e)}_{11}}^{u\cdot \gamma(u),v\cdot \gamma(v)}_{u\cdot \gamma(u),v\cdot \gamma(v)},\nonumber
\eqn
and
\beqn
\mathcal{P}_{u,v}:=\lrs{P}_{u\cdot \gamma(u),v\cdot \gamma(v)},\nonumber
\eqn
This gives the inversion
\beqn
\mathcal{P}_{u,v}&=e^{-\lambda}\exp\left(\eta_{u,v}\right),\\\nonumber
\eta_{u,v}&=\lambda+\ln\left(\mathcal{P}_{u,v}\right).
\eqn

Consider the $2^n\times 2^{n-1}$ rectangular matrices $F_{01}$, $F_{10}$ and $F_{11}$ with components 
\beqn
\begin{array}{ll}
\lrs{F_{01}}_{u,v}^e=\lrs{K^{(e)}_{01}}^{u,v}_{u,v}=(-1)^{|e\cap v_1|}, & \lrs{F_{10}}_{u,v}^e=\lrs{K^{(e)}_{10}}^{u,v}_{u,v}=(-1)^{|e\cap u_1|},\\
\lrs{F_{11}}_{u,v}^e=\lrs{K^{(e)}_{11}}^{u,v}_{u,v}=(-1)^{|e\cap u_1|+|e\cap v_1|};
\end{array}\nonumber
\eqn
where $e\subseteq \lrs{n-1}$ and $u=u_0\sep u_1$ and $v=v_0\sep v_1$ are ordered-bipartitions of $\lrs{n-1}$.
If we define the column vector $\vec{\eta}:=\{\eta_{u,v}\}$ indexed by pairs of ordered-bipartitions and the column vectors $\vec{\alpha}:=\{\alpha_e\}$, $\vec{\beta}:=\{\alpha_{e}\}$ and $\vec{\gamma}:=\{\alpha_{e}\}$ indexed by subsets of $\lrs{n-1}$, we then have the matrix equation
\beqn
\vec{\eta}=F_{01}\vec{\alpha}+F_{10}\vec{\beta}+F_{11}\vec{\gamma}.\nonumber
\eqn
Writing $H^{(n)}=H^{(n-1)}\otimes H$ with $H^{(1)}=H$, we note that
\beqn
\lrs{F_{01}}_{u,v}^e&=\lrs{H^{(n-1)}}_{u,v}^{\emptyset,e},\\\nonumber
\lrs{F_{10}}_{u,v}^e&=\lrs{H^{(n-1)}}_{u,v}^{e,\emptyset},\\
\lrs{F_{11}}_{u,v}^e&=\lrs{H^{(n-1)}}_{u,v}^{e,e};
\eqn
and define the $2^{n-1}\times 2^n$ rectangular matrices  $G_{01},G_{10}$ and $G_{11}$ as
\beqn
\lrs{G_{01}}_e^{u,v}&=\lrs{{H^{-1}}^{(n-1)}}_{\emptyset,e}^{u,v},\nonumber\\
\lrs{G_{10}}_e^{u,v}&=\lrs{{H^{-1}}^{(n-1)}}_{e,\emptyset}^{u,v},\\
\lrs{G_{11}}_e^{u,v}&=\lrs{{H^{-1}}^{(n-1)}}_{e,e}^{u,v}.
\eqn
Noting that
\beqn
\sum_{w,x}\lrs{{H^{-1}}^{(n-1)}}_{u,v}^{w,x}\lrs{H^{(n-1)}}_{w,x}^{y,z}=\delta_{u,y}\delta_{v,z},\nonumber
\eqn
for all $u,v,y,z$ ordered-bipartitions of $\lrs{n-1}$,  we then have the matrix identities
\beqn
G_{01} F_{01}=\id,\qquad G_{10} F_{10}=\id, \qquad G_{11} F_{11}=\id,\nonumber
\eqn
and
\beqn
G_{01}F_{10}=0=G_{01} F_{11}=G_\beta F_{11}=G_{10} F_{01}=G_{11} F_{01}=G_{11} F_{10}.\nonumber
\eqn

Writing
\beqn
\vec{\alpha}=G_{01} \vec{\eta},\qquad \vec{\beta}=G_{10}\vec{\eta}, \qquad \vec{\gamma}=G_{11}\vec{\eta},\nonumber
\eqn
completes the inversion for the K3ST model.

\section{Inversion of the $\mathbb{Z}_r$ model}\label{sec:Zr}
We now consider the group based model for $\mathbb{Z}_r=\left\{0,1,2,\ldots (r-1)\right\}_{+(\text{mod r})}\cong \lra{\sigma:\sigma^r=e}$.
For this model the generic rate matrix has the form
\beqn
Q=-\lambda\id +\sum_{i=1}^r\alpha^iK_{\sigma^i},\nonumber
\eqn
where $\lambda=\sum_{i=1}^r\alpha^i$ and
\beqn
K_\sigma=\left(\begin{array}{ccccc}
0 & 0 & \ldots & 0 & 1 \\
1 & 0 & 0 & 0 & 0 \\
0 & 1 & \ldots & 0 & 0 \\
\vdots & \vdots  & \ddots & \vdots & \vdots \\
0 & 0 & \ldots & 1 & 0 
\end{array}\right)\nonumber,
\eqn
so that $K_{\sigma^i}=K_{\sigma}^i$.

Defining $\omega=e^{2\pi i/r}$, we have $\omega^r=1$ and $1+\omega+\omega^2+\ldots +\omega^{r-1}=0$ and $\lrs{f}_{i}^{j}=\omega^{ij}$ where $i,j=0,1,2,\ldots,r-1$.
Of course, $f$ is the character table of $\mathbb{Z}_r$ and $\lrs{f^{-1}}^i_j=\fra{1}{r}\omega^{-ij}$.
\begin{lem}
\beqn\label{eq:orthogonality}
\sum_{\nu}\lrs{f\otimes f\otimes \ldots \otimes f}_{\mu}^{\nu}\lrs{f^{-1}\otimes f^{-1}\otimes \ldots \otimes f^{-1}}_{\nu}^{\mu'}=\delta_{\mu\mu'},\nonumber
\eqn
where $\mu,\nu,\mu'$ are ordered-$r$-partitions of the set $\lrs{n}$ corresponding to the strings $i_1i_2\ldots i_n$, $j_1j_2\ldots j_n$ and $k_1k_2\ldots k_n$.
\end{lem}
\begin{proof}
The result is obvious by the definition of tensor product. 
However, explicitly we have
\beqn
\sum_{\nu}\lrs{f\otimes f\otimes \ldots \otimes f}_{\mu}^{\nu}&\lrs{f^{-1}\otimes f^{-1}\otimes \ldots \otimes f^{-1}}_{\nu}^{\mu'}\\
&=\fra{1}{r^n}\sum_{0\leq j_1,j_2,\ldots ,j_{r-1}\leq (r-1)}\omega^{i_1j_1+i_2j_2+\ldots i_{r-1}j_{r-1}}\omega^{-\left(j_1k_1+j_2k_2+\ldots +j_nk_n\right)}\nonumber\\
&=\fra{1}{r^n}\sum_{0\leq j_1,j_2,\ldots ,j_{r-1}\leq (r-1)}\omega^{j_1(i_1-k_1)+j_2(i_2-k_2)+\ldots +j_n(i_n-k_n)}
\eqn
which clearly equals 1 if $i_{\ell}-k_{\ell}=0$ for all $\ell$, and, by repeatedly applying $1+\omega+\omega^2+\ldots +\omega^{r-1}=0$, equals 0 otherwise.
\end{proof}

The regular representation contains exactly one copy of every irreducible representation and the irreducible representations of $\mathbb{Z}_r$ are given by the powers of $\omega$:
\[\rho_{i}:\begin{array}{cc}\mathbb{Z}_r\rightarrow \mathbb{C} \\ \sigma\mapsto \omega^{i}\end{array}.\] 
Thus the change of basis $K_{\sigma^i}\mapsto \widehat{K}_{\sigma^i}=fK_{\sigma^i}f^{-1}$ will give diagonal matrices $\widehat{K}_{\sigma^i}$.
Additionally,
\begin{lem}
In the diagonal basis, the matrices $\widehat{K}_{\sigma^i}:=fK_{\sigma^i}f^{-1}$ have matrix elements $\lrs{\widehat{K}_{\sigma^s}}_{i}^j=\omega^{is}\delta_{ij}$.
\end{lem}
\begin{proof}
Consider the matrix elements $\lrs{K_{\sigma^s}}_{i}^{j}=\delta_{i\sigma^s(j)}$.
Thus
\beqn
\lrs{fK_{\sigma^s}f^{-1}}^i_j=\sum_{k,l}\omega^{ik}\delta_{k\sigma^s(l)}\omega^{-lj}=\sum_{l}\omega^{i\sigma^s(l)-lj}=\sum_{l}\omega^{i(l+s)-lj}&=\omega^{is}\sum_{l}\omega^{l(i-j)}=\omega^{is}\delta_{ij},\nonumber
\eqn
where we have used $\omega^{\sigma^s(m)}=\omega^{m+s}$.
\end{proof}

Now
\beqn
\lrs{\delta^{n-1}\pi}_{i_1i_2\ldots i_n}=\fra{1}{r}\delta_{i_1i_2}\delta_{i_1i_3}\ldots \delta_{i_1i_n},\nonumber
\eqn
and 
\beqn
\lrs{\widehat{\delta^{n-1}\pi}}_{i_1i_2\ldots i_n}&=\fra{1}{r}\sum_{j_1,j_2,\ldots ,j_r}\omega^{i_1j_1+i_2j_2+\ldots +i_nj_n}\delta_{j_1j_2}\delta_{j_1j_3}\ldots \delta_{j_1j_n}\nonumber\\
&=\fra{1}{r}\sum_{j_1}\omega^{j_1(i_1+i_2+\ldots +i_n)}\\
&=\left\{\begin{array}{ll} {1\text{ if }i_1+i_2+\ldots +i_n=0\text{ (mod r)}} \\ {0\text{, otherwise.}}\end{array}\right.
\eqn
Translating this result using the ordered-$r$-partitions for indices, we have
\begin{lem}\label{lem:Z4diagstar}
In the diagonal basis, the uniform initial distribution on the star tree has components
\beqn
\lrs{\widehat{\delta^{n-1}\pi}}_{\mu}=\left\{\begin{array}{ll} {1\text{ if }0|\mu_0|+1|\mu_1|+2|\mu_2|+\ldots +(r-1)|\mu_{r-1}|=0 \emph{ (mod r)}} \\ {0\text{, otherwise.}}\end{array}\right.,\nonumber
\eqn
where $\mu=\mu_0\sep \mu_1\sep \mu_2\sep \ldots \sep \mu_{r-1}$ is an ordered-$r$-partition of the set $\lrs{n}$.
\end{lem}

Again recall that for this model a generic phylogenetic tensor can be written as
\beqn
P=e^{-\lambda}\exp\left(\sum_{\emptyset \neq e\subseteq \lrs{n-1},s\in\lrs{r-1}  } \alpha_e^{s} K^{(e)}_{\sigma^s}\right)\delta^{n-1}\pi,\nonumber
\eqn
where $\pi=\fra{1}{r}(1,1,\ldots,1)^T$.
In the diagonal basis $\widehat{P}:=f^{(n)}\cdot P $ and as a consequence of Lemma~\ref{lem:Z4diagstar} $\widehat{P}$ will have many vanishing components.
To avoid these we take $u=u_0\sep u_1\sep u_2\sep \ldots \sep u_{r-1}$ as an \emph{ordered-$r$-partition} of $\lrs{n-1}$ and set
\beqn
\gamma(u)=r-\left(0|u_0|+1|u_1|+2|u_2|+\ldots +(r-1)|u_{r-1}|\right) \text{ (mod r)}.\nonumber
\eqn
If we define $\mathcal{P}_u:=\lrs{\widehat{P}}_{u\cdot \gamma(u)}$ and
\beqn
\eta_u:=\lrs{\sum_{\emptyset \neq e\subseteq \lrs{n-1},s\in\lrs{r-1}  } \alpha_e^{s} \widehat{K}^{(e)}_{\sigma^s}}_{u\cdot \gamma(u)}^{u\cdot \gamma(u)},\nonumber
\eqn
we then have the first part of the inversion for the $\mathbb{Z}_r$ model:
\beqn
\mathcal{P}_{u}&=e^{-\lambda}\exp\left(\eta_{u}\right),\\\nonumber
\eta_{u}&=\ln\left(\mathcal{P}_{u}\right)+\lambda.
\eqn

For each $i\in\lrs{r-1}$, we define the column vectors $\vec{\alpha}_i :=\left\{\alpha_{e}^{i}\right\}_{\emptyset \neq e\subseteq \lrs{n-1}}$, and, for each $\emptyset \neq e\subseteq \lrs{n-1}$ and $u$ an ordered-$(r-1)$-partition of $\lrs{n-1}$, we define the rectangular $r^{n-1}\times 2^{n-1}$ matrices
\beqn
\begin{array}{llll}
\lrs{F_1}_{u}^e:=\lrs{K^{(e)}_{\sigma}}_{u\cdot \gamma(u)}^{u\cdot \gamma(u)}, & \lrs{F_2}^{e}_u:=\lrs{K^{(e)}_{\sigma^2}}^{u\cdot \gamma(u)}_{u\cdot \gamma(u)}, & \ldots & \lrs{F_{r-1}}_{u}^e:=\lrs{K^{(e)}_{\sigma^{r-1}}}^{u\cdot \gamma(u)}_{u\cdot \gamma(u)},\nonumber
\end{array}
\eqn
so we have the vector equation
\beqn
\eta=F_1\vec{\alpha_1}+F_2\vec{\alpha_2}+\ldots + F_{r-1}\vec{\alpha}_{r-1}.\nonumber
\eqn

We claim that
\begin{lem}\label{lem:blah}
\beqn
\lrs{F_1}^{e}_u&=\lrs{f^{(n-1)}}^{e^c\sep e\sep \emptyset\sep \emptyset\sep \ldots\sep \emptyset}_{ u}, \\
\lrs{F_2}^{e}_u&=\lrs{f^{(n-1)}}^{e^c \sep \emptyset\sep e\sep \emptyset\sep \ldots\sep \emptyset}_{u},\\
&\vdots\\
\lrs{F_{r-1}}^{e}_u&=\lrs{f^{(n-1)}}^{ e^c\sep \emptyset\sep \emptyset\sep \emptyset\sep \ldots\sep e}_{u}.\\
\eqn
\end{lem}
\begin{proof}
We recall that $\lrs{\widehat{K}_{\sigma^s}}_i^j=\omega^{is}\delta_{ij}$, so, for $\mu=\mu_0\sep \mu_1\sep \mu_2\sep \ldots \sep \mu_{r-1}$ an ordered-$r$-parition of $\lrs{n}$, and $e$ a subset of $\lrs{n-1}$ we have
\beqn
\lrs{\widehat{K}^{(e)}_{\sigma^s}}^\mu_\mu=\omega^{s\left(0|\mu_0\cap e|+1|\mu_1\cap e|+\ldots +(r-1)|\mu_{r-1}\cap e|\right)},\nonumber
\eqn
so
\beqn
\lrs{\widehat{K}^{(e)}_{\sigma^s}}^{u\cdot \gamma(u)}_{u\cdot \gamma(u)}=\omega^{s\left(0|u_0\cap e|+1|u_1\cap e|+\ldots +(r-1)|u_{r-1}\cap e|\right)},\nonumber
\eqn
because $e\subseteq \lrs{n-1}$.
On the other hand $\lrs{f}_i^j=\omega^{ij}$, so
\beqn
\lrs{f^{(n-1)}}_{u}^{e^c\sep \emptyset\sep \ldots \sep \emptyset\sep e\sep \emptyset \sep \ldots \sep \emptyset}=\omega^{s\left(0|u_0\cap e|+1|u_1\cap e|+\ldots +(r-1)|u_{r-1}\cap e|\right)},\nonumber
\eqn
where $e$ appears in the $s^{th}$ position.
\end{proof}

Define, for $i\in \lrs{r-1}$, the rectangular $2^{n-1}\times r^{n-1}$ matrices
\beqn
\lrs{G_1}^u_e:&=\lrs{{f^{-1}}^{(n-1)}}^{e^c\cdot \gamma(u)\sep e\sep \emptyset\sep \emptyset\sep \ldots\sep \emptyset}_{ u\cdot \gamma(u)}\\
\lrs{G_2}^u_e:&=\lrs{{f^{-1}}^{(n-1)}}^{e^c\cdot \gamma(u)\sep \emptyset\sep e\sep \emptyset\sep \ldots\sep \emptyset}_{ u\cdot \gamma(u)}\\
&\vdots\\
\lrs{G_{r-1}}^{u}_e&=\lrs{{f^{-1}}^{(n-1)}}^{ e^c\cdot\gamma(u)\sep \emptyset\sep \emptyset\sep \emptyset\sep \ldots\sep e}_{u\cdot \gamma(u)}.\nonumber
\eqn
Of course $G_iF_j=\delta_{ij}\id$, so we now have the second part of the inversion:
\beqn
\vec{\alpha_i}=G_{i}\eta.\nonumber
\eqn

\section{Inversion of any abelian group-based model}\label{sec:anyabelian}

\begin{lem}\label{lem:abeliangroup}
Any (finitely generated) abelian group $G$ is isomorphic to a direct product of cyclic groups of prime-power order, ie. $G\cong \mathbb{Z}_{r_1}\times \mathbb{Z}_{r_2}\times \ldots \times \mathbb{Z}_{r_q}$ where each $r_i=p_i^{n_i}$ where $p_i$ is prime and $n_i$ is a positive integer.
\end{lem}


\begin{lem}
The group-based model arising from the $G$ is defined only up to group isomorphisms of $G$.
\end{lem}
\begin{proof}
A generic rate matrix for the group-based model arsing from $G$ is given by
\beqn
Q=-\lambda\id+\sum_{e\neq \sigma\in G}\alpha^{\sigma}K_{\sigma}.\nonumber
\eqn
Under a group isomorphism $\phi:G\rightarrow G'$, we have $\phi(\sigma_i\sigma_j)=\phi(\sigma_i)\phi(\sigma_j)$.

Recall (\ref{eq:Ksigdef}), so that the matrix elements $\lrs{K_{\sigma}}_i^j$ is set via the action $\sigma_i\mapsto \sigma\sigma_i=\sigma_j$.
If we consider the regular representation of $G'$ we then have $\lrs{K_{\phi(\sigma})}_i^j$ defined by $\phi(\sigma_i)\mapsto \phi(\sigma)\phi(\sigma_i)$.
Now $\phi(\sigma)\phi(\sigma_i)=\phi(\sigma\sigma_i)=\phi(\sigma_j)$ and, because $\phi$ is a group isomorphism, this occurs if and only if $\sigma\sigma_i=\sigma_j$.
Thus $\lrs{K_{\phi(\sigma})}_i^j=\lrs{K_{\sigma}}_i^j$ for all $i$ and $j$.
\end{proof}

This means that we can restrict attention to a single representitive in the isomorphism class of $G$.
Of course, for this purpose we choose the representative guaranteed by Lemma~\ref{lem:abeliangroup}.

Thus, for any abelian group $G$, with generators $\sigma_1,\sigma_2,\ldots ,\sigma_q$, as per Lemma~\ref{lem:abeliangroup}, the corresponding group-based model has rate generators given by
\beqn
L_{\sigma}=-\id+{K}_{\sigma_1^{m_1}}\otimes {K}_{\sigma_2^{m_2}}\otimes \ldots \otimes {K}_{\sigma_q^{m_q}},\nonumber
\eqn
for all $e\neq \sigma=(\sigma_1^{m_1},\sigma_2^{m_2},\ldots,\sigma_{q}^{m_q})\in G$, where $K_{\sigma_i}$ is the permutation matrix representing the generator $\sigma_i\in \mathbb{Z}_{r_i}$.
The character table $f$ of $G$ is simply the tensor product of the individual character tables of the $\mathbb{Z}_{r_i}$:
\beqn
f=f_1\otimes f_2\otimes \ldots \otimes f_q.\nonumber
\eqn
In the diagonal basis we have matrix elements
\beqn
\lrs{f_kK_{\sigma_k^{s}}f_k^{-1}}_i^j=\lrs{\hat{K}_{\sigma_k^{s}}}_i^j=\left(\omega_{k}\right)^{is}\delta_{ij},\nonumber
\eqn
where $\omega_k$ is a $k^{th}$ root of unity.
Thus
\beqn
\lrs{\hat{K}_{\sigma_1^{m_1}}\otimes \hat{K}_{\sigma_2^{m_2}} \otimes \ldots \hat{K}_{\sigma_q^{m_q}}}_{i_1i_2\ldots i_q}^{j_1j_2\ldots j_q}=\delta_{i_1j_1}\delta_{i_2j_2}\ldots \delta_{i_qj_q}\left(\omega_{1}\right)^{i_1m_1}\left(\omega_{2}\right)^{i_2m_2}\ldots \left(\omega_{q}\right)^{i_qm_q}.\nonumber
\eqn

We write phylogenetic tensors for this model in the form $P_{i_{11}i_{12}\ldots i_{1n},i_{21}i_{22}\ldots i_{2n}\ldots \ldots i_{q1}i_{q2}\ldots i_{qn}}$, where $0\leq i_{sj}\leq r_s$ for all $0\leq s\leq q$.
We simplify notation by writing each group of indices as $\mu^{(s)}:=i_{s1}i_{s2}\ldots i_{sn}$ where $\mu^{(s)}$ is an ordered-$r_s$-partition of $\lrs{n}$.

\begin{lem}
In the diagonal basis, the uniform initial distribution on the star tree has components
\beqn
\lrs{\widehat{\delta^{n-1}\pi}}_{\mu^{(1)}\mu^{(2)}\ldots \mu^{(q)}}=\left\{\begin{array}{ll}1,\text{ if, }0|\mu_0^{(i)}|+1|\mu_{1}^{(i)}|+\ldots +(r_i-1)|\mu_{r-1}^{(i)}|=0,\forall i;\\
0,\text{ otherwise.}\end{array}\right.\nonumber
\eqn
\end{lem}

A generic phylogenetic tensor for this model can be expressed as
\beqn
P=e^{-\lambda}\exp\left(\sum_{\emptyset \neq e\subseteq \lrs{n-1},s_i\in\lrs{r_i-1}}\alpha^{s_1s_2\ldots s_q}_{e}K^{(e)}_{\sigma^{s_1}_1}\otimes K^{(e)}_{\sigma^{s_2}_2}\otimes \ldots \otimes K^{(e)}_{\sigma^{s_q}_q}\right)\cdot \delta^{n-1}\pi,
\nonumber
\eqn
where $\pi$ is the unifrom distribution on ${\sum_{i=1}^q r_i}$ states, i.e.  $\pi=({\sum_{i=1}^q r_i})^{-1}(1,1,\ldots ,1)^T$.

In the diagonal basis $\widehat{P}=(f_1\otimes f_2\otimes\ldots \otimes f_q)^{(n)}\cdot P$, and, as a consequence of the previous lemma, $P$ has many vanishing components. 
To avoid these, for each $i\in \lrs{q}$ we take $u^{(i)}=u^{(i)}_0\sep u^{(i)}_1\sep u^{(i)}_2\sep \ldots \sep u^{(i)}_{r_i-1} $ as an \emph{ordered-$r_i$-partition} of $\lrs{n-1}$ and set
\beqn
\gamma_{i}(u^{(i)})=r_i-(0|u^{(i)}_0|+1|u^{(i)}_1|+2|u^{(i)}_2|+\ldots +(r_i-1)|u^{(i)}_{r-1}|) \text{ (mod $r$)}.\nonumber
\eqn

We then define
\beqn
\mathcal{P}_{u^{(1)}u^{(2)}\ldots u^{(q)}}:=\lrs{\widehat{P}}_{u^{(1)}\cdot \gamma_1(u^{(1)})u^{(2)}\cdot \gamma_2(u^{(2)})\ldots u^{(q)}\cdot \gamma_1(u^{(q)})},\nonumber
\eqn
and
\beqn
\eta_{u^{(1)}u^{(2)}\ldots u^{(q)}}:=\lrs{\sum_{\emptyset \neq e\subseteq \lrs{n-1},s_i\in\lrs{r_i-1}}\alpha^{s_1s_2\ldots s_q}_{e}\widehat{K}^{(e)}_{\sigma^{s_1}_1}\otimes \widehat{K}^{(e)}_{\sigma^{s_2}_2}\otimes \ldots \otimes \widehat{K}^{(e)}_{\sigma^{s_q}_q}}^{u^{(1)}\cdot \gamma_1(u^{(1)})u^{(2)}\cdot \gamma_2(u^{(2)})\ldots u^{(q)}\cdot \gamma_1(u^{(q)})}_{u^{(1)}\cdot \gamma_1(u^{(1)})u^{(2)}\cdot \gamma_2(u^{(2)})\ldots u^{(q)}\cdot \gamma_1(u^{(q)})},\nonumber
\eqn
so that we have the first part of the inversion
\beqn
\mathcal{P}_{u^{(1)}u^{(2)}\ldots u^{(q)}}&=e^{-\lambda}\exp\left(\eta_{u^{(1)}u^{(2)}\ldots u^{(q)}}\right),\\
\eta_{u^{(1)}u^{(2)}\ldots u^{(q)}}&=\lambda+\ln\left(\mathcal{P}_{u^{(1)}u^{(2)}\ldots u^{(q)}}\right).
\nonumber
\eqn

We define the column vectors $\vec{\alpha}^{s_1s_2\ldots s_q}:=\{\alpha^{s_1s_2\ldots s_q}_e\}_{\emptyset\neq e\subseteq \lrs{n-1}}$ and $\vec{\eta}:=\{\eta_{u^{(1)}u^{(2)}\ldots u^{(q)}}\}$ where $u_i$ is an ordered-$r_i$-partition of $\lrs{n-1}$, and the $ (r_1r_2\ldots r_q)^{n-1}\times 2^{n-1}$ matrices
\beqn
\lrs{F_{s_1s_2\ldots s_q}}^e_{u_1u_2\ldots u_q}:&=\lrs{K_{\sigma^{s_1}_1}^{(e)}}^{u_1\cdot \gamma(u_1)}_{u_1\cdot \gamma(u_1)}\lrs{K_{\sigma^{s_2}_2}^{(e)}}^{u_2\cdot \gamma(u_2)}_{u_2\cdot \gamma(u_2)}\ldots\lrs{K_{\sigma^{s_q}_q}^{(e)}}^{u_q\cdot \gamma(u_q)}_{u_q\cdot \gamma(u_q)}\\
&=\lrs{f_1^{(n-1)}}_{u_1}^{e^c\sep \emptyset \sep \ldots \sep \emptyset \sep e\sep \emptyset \sep \ldots \sep\emptyset}
\lrs{f_2^{(n-1)}}_{u_2}^{e^c\sep \emptyset \sep \ldots \sep \emptyset \sep e\sep \emptyset \sep \ldots \sep\emptyset}
\ldots
\lrs{f_q^{(n-1)}}_{u_q}^{e^c\sep \emptyset \sep \ldots \sep \emptyset \sep e\sep \emptyset \sep \ldots \sep\emptyset},
\nonumber
\eqn
where in each term $e$ appears in the $s_i^{th}$ position and the equality follows from Lemma~\ref{lem:blah}.

We can then write the vector equation
\beqn
\vec{\eta}=\sum_{s_1s_2\ldots s_q:1\leq s_i\leq r_i-1}F_{s_1s_2\ldots s_q}\vec{\alpha}^{s_1s_2\ldots s_q}.\nonumber
\eqn
If we define the $ 2^{n-1} \times(r_1r_2\ldots r_q)^{n-1}$ matrices
\beqn
\lrs{G_{s_1s_2\ldots s_q}}_e^{u_1u_2\ldots u_q}
=\lrs{{f_1^{-1}}^{(n-1)}}_{u_1}^{e^c\sep \emptyset \sep \ldots \sep \emptyset \sep e\sep \emptyset \sep \ldots \sep\emptyset}
\lrs{{f_2^{-1}}^{(n-1)}}_{u_2}^{e^c\sep \emptyset \sep \ldots \sep \emptyset \sep e\sep \emptyset \sep \ldots \sep\emptyset}
\ldots
\lrs{{f_q^{-1}}^{(n-1)}}_{u_q}^{e^c\sep \emptyset \sep \ldots \sep \emptyset \sep e\sep \emptyset \sep \ldots \sep\emptyset},
\nonumber
\eqn
where in each term $e$ appears in the $s_i^{th}$ position, we have the orthogonality relations
\beqn
G_{s_1s_2\ldots s_q}F_{s_1's_2'\ldots s_q'}=\delta_{s_1s_1'}\delta_{s_2s_2'}\ldots \delta_{s_qs_q'}\id.\nonumber
\eqn

This gives us the second part of the inversion of any group-based model:
\beqn
\vec{\alpha}^{s_1s_2\ldots s_q}=G_{s_1s_2\ldots s_q}\vec{\eta}.\nonumber
\eqn

\section{Conclusion}

In this article we have given an alternative derivation of the inversion of group-based phylogenetic models.
Primarily our method relies on the remarkable intertwining relation between branching events and Markov evolution (\ref{eq:intertwine}), and the resulting simplified expression of phylogenetic tensors given in (\ref{eq:networkrep}).
From there we took a representation theoretic approach concentrating on the structure of tensor indices.

\section*{Funding}

This research was conducted with support from Australian Research Council (ARC) Discovery Grant DP0770991 and Future Fellowship FT100100031.

\bibliographystyle{jtbnew}
\bibliography{masterC}

\end{document}